\documentclass[11pt,twoside]{amsart}
\usepackage{amsmath}
\usepackage{amssymb}
\usepackage{bm}
\usepackage{ascmac}
\usepackage{braket}
\usepackage{color}
\usepackage{comment}
\usepackage{amsthm}
\usepackage[dvipdfmx]{graphicx}
\usepackage{tikz-cd}
\usepackage{fancyhdr}

\newtheorem{thm}{Theorem}[section]

\newtheorem{lem}{Lemma}[section]

\theoremstyle{definition}
\newtheorem{dfn}{Definition}[section]

\theoremstyle{remark}

\numberwithin{equation}{section}

\begin{document}

\title{Energy functions of general dimensional diamond crystals based on the Kitaev model}
\author{Akito Tatekawa}
\address{}
\email{akitot@meiji.ac.jp}
\date{}
\maketitle
\setlength{\headheight}{14pt}

\begin{abstract}
The purpose of this paper is to extend the Kitaev model to a general dimensional diamond crystal. 
 We define the Hamiltonian by using representations of Clifford algebras. 
Then we compute the energy functions. 
We show that the energy functions are  identified with those appearing in the tight binding model. 
\end{abstract}

\section{Introduction}
The Kitaev model is an exactly solvable model of a spin $\frac{1}{2}$ system on the honeycomb lattice. 
This model was extensively studied by A.~Kitaev in \cite{Kitaev}. 
The purpose of this paper is to extend the Kitaev model to the $d$-dimensional diamond crystal $\Delta_d$ for any $d\geq 2$ and to compute the energy functions. 
The $d$-dimensional diamond crystal $\Delta_d$ was defined by T.~O'Keeffe \cite{Keeffe}. 
The honeycomb lattice can be treated as a two-dimensional diamond crystal $\Delta_2$. 
In the case where $d=3$, $\Delta_3$ is the diamond crystal in $\mathbb{R}^3$. 
The Kitaev model of $\Delta_3$ was investigated by S.~Ryu \cite{Shinsei}. 


We define the Majorana operators  by using  irreducible representations of the Clifford algebras. 
The creation operators $a_i^{\dagger}$ and  the annihilation operators $a_i$ are defined by means of the Majorana operators. 
There is an action of the root lattice of type $A_d$ on $\Delta_d$ and the quotient space is a graph denoted by $X_0$. 
We call $X_0$ the base graph of the diamond crystal $\Delta_d$. 
The base graph was studied by T.~Sunada \cite{Sunada} in the framework of topological crystallography. 
We effectively use the base graph $X_0$ to describe the Hamiltonian for $\Delta_d$. 
We compute the energy functions of the Kitaev model of $\Delta_d$ by applying a method based on  the Fourier transform developed by Kato and Richard \cite{Kato}. 

The energy functions for crystal lattices in quantum mechanics are described by the Schr\"{o}dinger equation with a periodic potential. 
However, it is difficult to solve the Schr\"{o}dinger equation analytically. 
We consider the tight binding model using the base graph of $\Delta_d$. 
Then, we compare the tight binding model and the Kitaev model of $\Delta_d$.

The paper is organized in the following way. 
In section 2, we recall the definition of the $d$-dimensional diamond crystal $\Delta_d$. 
In section 3, we review the Kitaev model for the honeycomb lattice $\Delta_2$. 
In section 4, we define the space of states  on which the Majorana operators act. 
In section 5, we formulate the Hamiltonian of the Kitaev model of  $\Delta_d$. 
In section 6, we compute the energy functions of  $\Delta_d$ by applying the Fourier transform. 
In section 7, we describe  zeros of the energy functions and energy gaps. 
In section 8, we identify  the energy functions of  the Kitaev model with those appearing in  the tight binding model.

\section{d-dimensional diamond crystal $\Delta_d$}
Following T.~Sunada \cite{Sunada}, we recall the definition of the $d$-dimensional diamond crystal $\Delta_d$ (see also T.~O'Keeffe \cite{Keeffe}). 

We set 
\begin{align*}
    W=\{\sum_{i=1}^{d+1}y_i\bm{e}_i\in\mathbb{R}^{d+1}\mid \sum_{i=1}^{d+1}y_i=0,\ y_i\in \mathbb{R}\}. 
\end{align*}
We define the root lattice $A_d$ as 
\begin{align}
    \label{eq:root}
    A_d=\{\sum_{i=1}^{d}n_i \alpha_i\in W\mid \alpha_i=\bm{e}_i-\bm{e}_{d+1},\ n_i\in\mathbb{Z}\}. 
\end{align}
We set $p=\frac{1}{d+1}\sum_{i=1}^{d}\alpha_i$ and
\begin{align*}
    A_d+p=\{p+\sum_{i=1}^{d}n_i\alpha_i\in W\mid \alpha_i=\bm{e}_i-\bm{e}_{d+1},\ n_i\in\mathbb{Z}\}. 
\end{align*}
Here, we denote the standard basis of $\mathbb{R}^{d+1}$ by $\{\bm{e}_i\}_{1\leq i\leq d+1}$. 
    \begin{dfn}
        We define the $d$-dimensional diamond crystal denoted by $\Delta_d$ as a spatial graph in the following way. 
        \label{dfn1}
        
        (1) The set of vertices of $\Delta_d$ is defined as the disjoint union $V(\Delta_d)=A_d\sqcup (A_d+p)$. 

        (2) The set of edges $E(\Delta_d)$ consists of the segments connecting $a' \in A_d+p$ and $a'-p \in A_d$, and the segments connecting $a' \in A_d+p$ and $a'+\alpha_i -p \in A_d$ for $1\leq i\leq d$. 
        We suppose that the edges of $E(\Delta_d)$ are unoriented. 
    \end{dfn}
From (1) and (2), it follows that $\Delta_d$ is a bipartite graph. 
The lattice group $\Gamma_{A_d}$ is generated by the translations $t_{\alpha_i}$ for $1\leq i\leq d$, where  $t_{\alpha_i}$ is defined by $t_{\alpha_i}(x)=x+\alpha_i$ for $x\in \mathbb{R}^d$. 
For example, in the case $d=2$, the $2$-dimensional diamond crystal $\Delta_2$ is the honeycomb lattice, and in the case $d=3$, $\Delta_3$ is the $3$-dimensional diamond crystal as shown in Figure \ref{fig:diamond}. 
There are $d+1$ edges meeting at each vertex of $\Delta_d$. 

\begin{figure}
    \begin{center}
    \includegraphics[width=140mm]{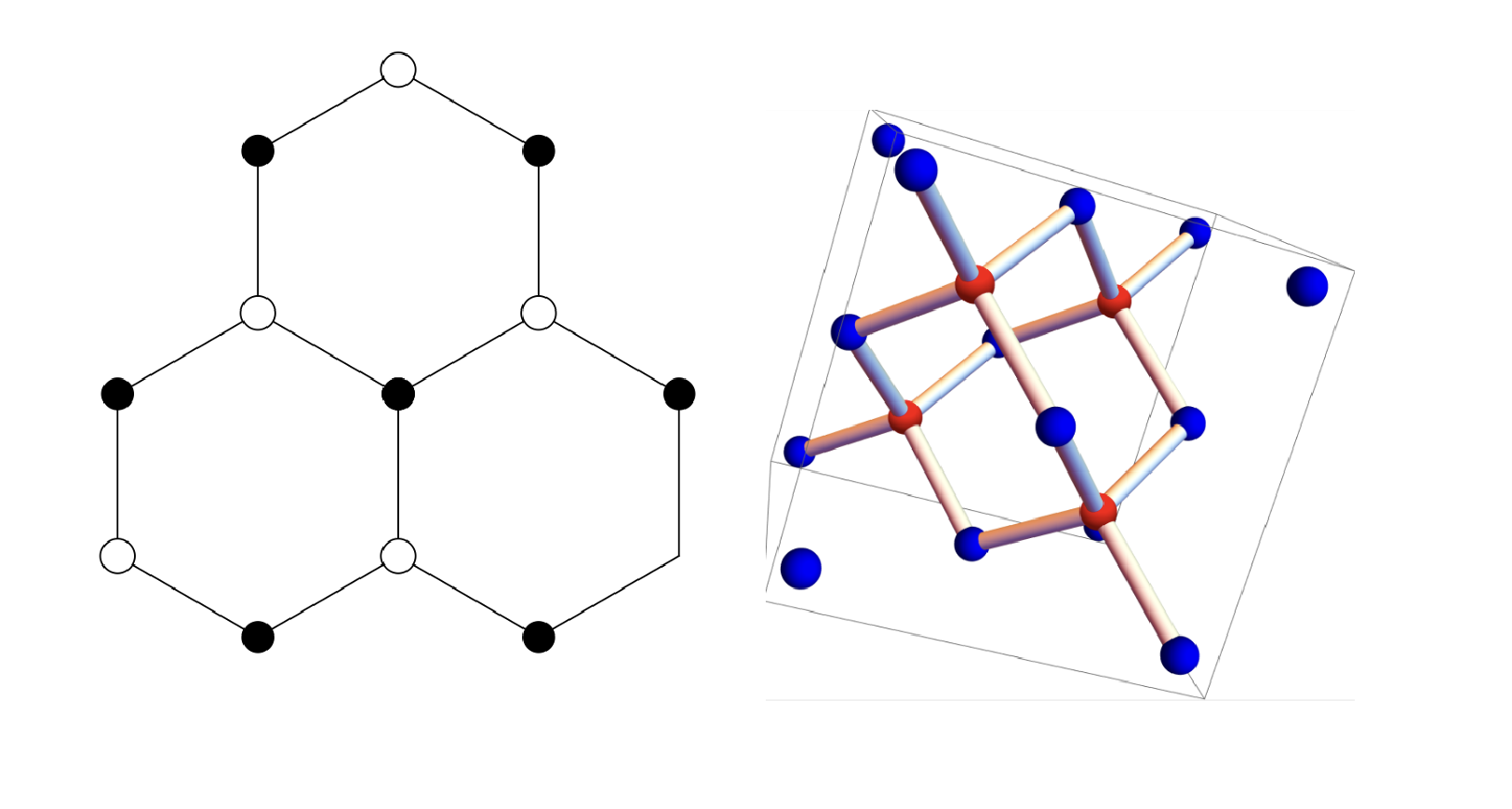}
    \caption{The honeycomb lattice $\Delta_2$ and the diamond crystal $\Delta_3$}
    \label{fig:diamond}
    \end{center}
\end{figure}

\section{The Kitaev model}
We review the definition of the Kitaev model following the article \cite{Kitaev}. 
It is a statistical mechanics model on the honeycomb lattice. 
There are three directions of edges meeting at each vertex. 
We call these directions $x$-link, $y$-link, and $z$-link as shown in Figure \ref{fig:links}. 
Let $V$ be a 2-dimensional vector space over $\mathbb{C}$ with basis $|0\rangle$ and $|1\rangle$. 
We set $\widetilde{M} = V\otimes V$. 
\begin{figure}[htbp]
    \begin{center}
    \includegraphics[width=100mm]{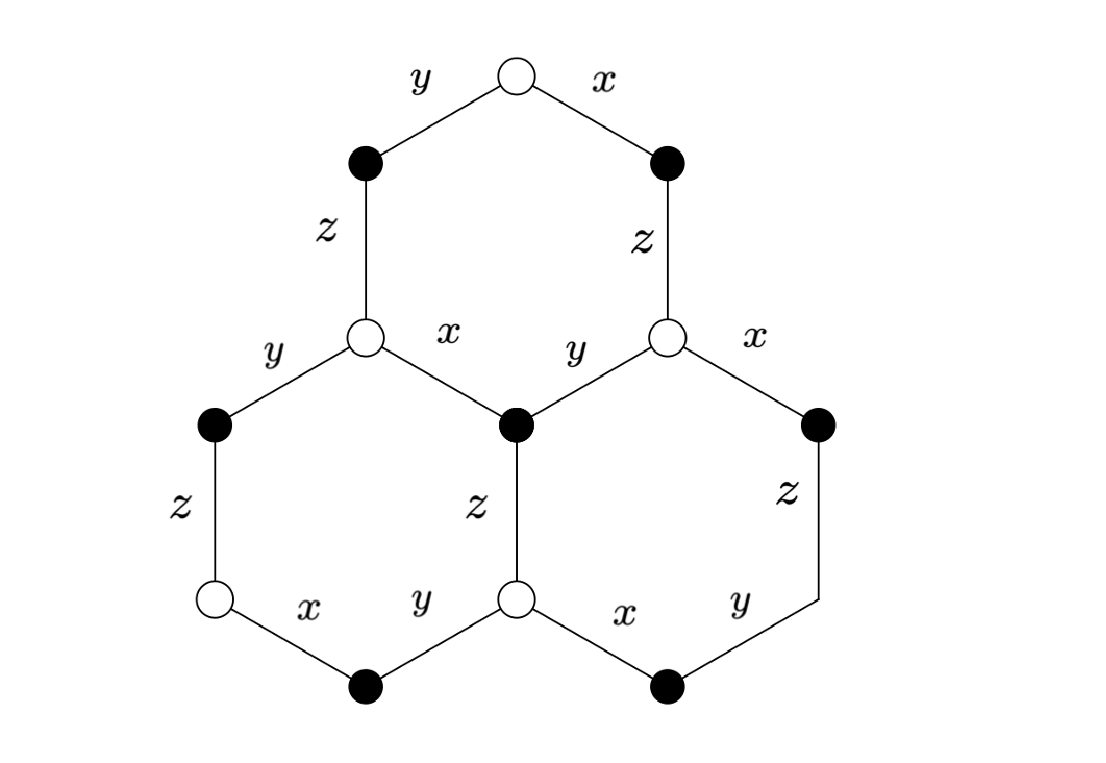}
    \caption{Links of the honeycomb lattice}
    \label{fig:links}
    \end{center}
\end{figure}
We set $|ij\rangle=|i\rangle\otimes|j\rangle$, for $i,\ j=0,\ 1$. 
We define the creation operators and the annihilation operators $a_1^{\dagger}$, $a_2^{\dagger}$, $a_1$, and $a_2$  
by matrices as bellow. 
We assume that $a_1|00\rangle=0$ and  $a_2|00\rangle=0$. 
We define the operators $a_1^{\dagger}$, $a_2^{\dagger}$, $a_1$, and $a_2$ by
\begin{align*}
    a_1=
    \begin{pmatrix}
        0&1&0&0\\
        0&0&0&0\\
        0&0&0&1\\
        0&0&0&0
    \end{pmatrix},\
    a_2=
    \begin{pmatrix}
        0&0&1&0\\
        0&0&0&-1\\
        0&0&0&0\\
        0&0&0&0
    \end{pmatrix},\\
    a_1^{\dagger}=
    \begin{pmatrix}
        0&0&0&0\\
        1&0&0&0\\
        0&0&0&0\\
        0&0&1&0
    \end{pmatrix},\
    a_2^{\dagger}=
    \begin{pmatrix}
        0&0&0&0\\
        0&0&0&0\\
        1&0&0&0\\
        0&-1&0&0
    \end{pmatrix}. 
\end{align*}
with respect to the basis $|00\rangle, |01\rangle, |10\rangle, |11\rangle$. 
The operators $a_i^{\dagger}$ and $a_i$ satisfy the anticommutation relations
    \begin{align*}
        \{a_i,a_j\}=\{a_i^{\dagger},a_j^{\dagger}\}=0,\ \{a_i,a_j^{\dagger}\}=\delta_{ij}
    \end{align*}
    where $\{x,y\}=xy+yx$. 
    The Majorana operators  $c_1, c_2, c_3, c_4$ are defined as
    \begin{align*}
    c_1=a_1+a_1^{\dagger},\ \ c_2=\frac{1}{\sqrt{-1}}(a_1-a_1^{\dagger}),\ \ c_3=a_2+a_2^{\dagger},\ \ c_4=\frac{1}{\sqrt{-1}}(a_2-a_2^{\dagger}). 
\end{align*}
In section 4, we define the Majorana operators  for any dimensions systematically by using the Clifford algebras. 

The spin operators $\sigma^{x},\sigma^{y}$ and $\sigma^{z}$  are defined as 
\begin{align*}
    \sigma^{x}=-\sqrt{-1}c_{1}c_{4},\ \ \sigma^{y}=\sqrt{-1}c_{2}c_{4},\ \ \sigma^{z}=-\sqrt{-1}c_{3}c_{4}. 
\end{align*}
To each vertex $v$ of $\Delta_2$ we associate the above $\widetilde{M}$ and denote it by $\widetilde{M}_v$. 
The operator $c_4^v$ is the action of $c_4$ on $\widetilde{M}_v$ and  $Id$ on the other components of $\bigotimes_{v\in V(\Delta_2)}\widetilde{M}_v$. 
To three directions $x$, $y$, and $z$ of the edges meeting at $v$ we associate the operators $c_1^v$, $c_2^v$, and $c_3^v$, which are the action of $c_1$, $c_2$ and $c_3$ on $\widetilde{M}_v$ and  $Id$ on the other components  of $\bigotimes_{v\in V(\Delta_2)}\widetilde{M}_v$  (see Figure \ref{fig:Majoran}). 
The spin operators $\sigma^{x}_v,\sigma^{y}_v$, and $\sigma^{z}_v$  are defined as 
\begin{align*}
    \sigma^{x}_v=-\sqrt{-1}c_{1}^vc_{4}^v,\ \ \sigma^{y}_v=\sqrt{-1}c_{2}^vc_{4}^v,\ \ \sigma^{z}_v=-\sqrt{-1}c_{3}^vc_{4}^v. 
\end{align*}
\begin{figure}[htbp]
\begin{center}
\includegraphics[width=120mm]{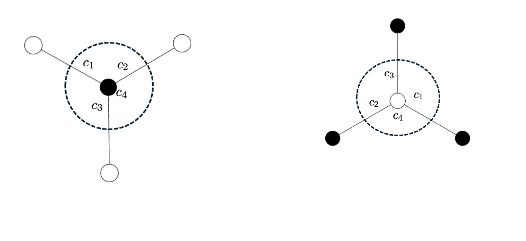}
\caption{Majorana operators of the honeycomb lattice}
\label{fig:Majoran}
\end{center}
\end{figure}
The action of $D$ on the space $\widetilde{M}$ is defined by $D=-c_1c_2c_3c_4$. 
We define the operator $\widetilde{D}$ acting on the space $\bigotimes_{v\in V(\Delta_2)}\widetilde{M}_v$ as
\begin{align*}
   &\widetilde{D}(\bigotimes_{v\in V(\Delta_2)}u_v)=(\bigotimes_{v\in V(\Delta_2)}Du_v),\ u_v\in \widetilde{M}_v. 
\end{align*}
The subspace $M\subset\widetilde{M}$ is defined by
\begin{align*}
   M= \{u\in \widetilde{M}\mid Du=u\}. 
\end{align*}
Then, $M$ is a 2-dimensional vector space over $\mathbb{C}$. 
The subspace $M(\Delta_2)\subset\bigotimes_{v\in V(\Delta_2)}\widetilde{M}_{v_u}$ is defined by
\begin{align*}
   M(\Delta_2)= \{u\in \bigotimes_{v\in V(\Delta_2)}\widetilde{M}_{v} \mid \widetilde{D}u=u\}. 
\end{align*}

The linear transformations $\sigma^x,\ \sigma^y$, and $\sigma^z$ are expressed by the Pauli spin matrices as 
\begin{align*}
    \sigma^x= 
    \begin{pmatrix}
        0&1\\
        1&0
    \end{pmatrix},\
    \sigma^y= 
    \begin{pmatrix}
        0&-\sqrt{-1}\\
        \sqrt{-1}&0
    \end{pmatrix},\
    \sigma^z=
    \begin{pmatrix}
        1&0\\
        0&-1
    \end{pmatrix}
\end{align*}
with respect to the basis $|00\rangle$ and $|11\rangle$ of $M$. 
We define $E_x$ as the set of unoriented edges of $\Delta_2$ in the direction $x$-link. 
For $y$ link and $z$  link we define $E_y$ and $E_z$ in the same way by replacing $x$-link with $y$-link and $z$-link respectively. 

The Kitaev model is defined by the Hamiltonian
$$
    H=-J_x\sum_{(v,v')\in E_x(\Delta_2)}\sigma_{v}^{x}\sigma_{v'}^{x}-J_y\sum_{(v,v')\in E_y(\Delta_2)}\sigma_{v}^{y}\sigma_{v'}^{y}-J_z\sum_{(v,v')\in E_z(\Delta_2)}\sigma_{v}^{z}\sigma_{v'}^{z}\\
$$
where $J_x, J_y, J_z\in \mathbb{R}$. 
In this way, the model can be identified with the Kitaev honeycomb model formulated in terms of Pauli spin matrices. 
Then the Hamiltonian $H$ acts on $M(\Delta_2)$. 
\begin{figure}[htbp]
\begin{center}
\includegraphics[width=70mm]{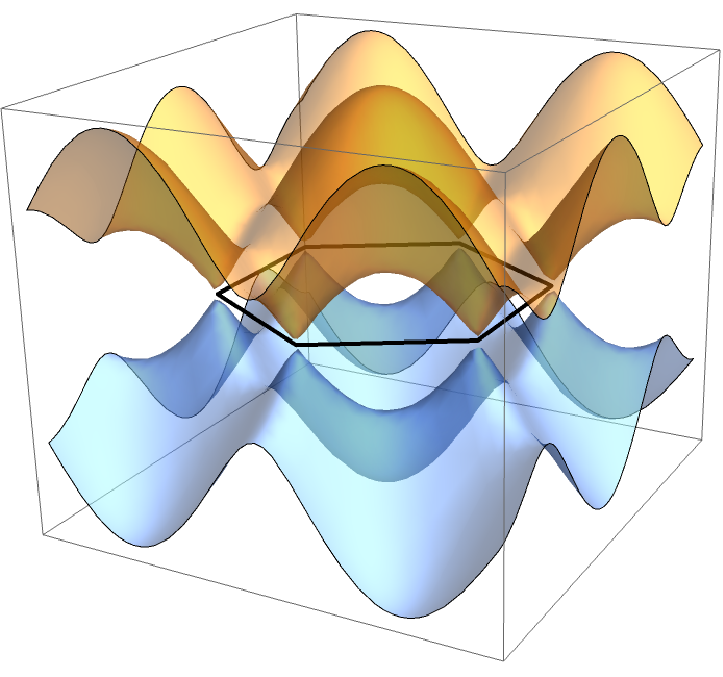}
\caption{Minimum ground state energy functions of the honeycomb lattice}
\label{fig:en}
\end{center}
\end{figure}
For the $\mathbb{Z}$-basis $\alpha_1, \alpha_2$ of $A_2$, we choose the vectors $b_1,\ b_2$ such that $(b_i,\alpha_j)=2\pi\delta_{ij}$ where $(\ ,\ )$ is the Euclidean inner product. 
We set $\bm{q}=k_1b_1+k_2b_2, \ k_1,\ k_2\in \mathbb{R}$. 

The minimum ground state energy functions of the Kitaev model are expressed as
$$
    \xi(\bm{q})=\pm|f(\bm{q})|
$$
with
$$
f(\bm{q})=2(J_xe^{\sqrt{-1}(\bm{q},\alpha_1)}+J_ye^{\sqrt{-1}(\bm{q},\alpha_2)}+J_z). \\
$$
The graph of these spectra as functions  in $q$ is shown in Figure \ref{fig:en} when the parameters satisfy $J_x=J_y=J_z=J$.

\section{Representations of Clifford algebras and $\Delta_d$}

The Clifford  algebra  $Cl_k$ is an algebra over $\mathbb{C}$  generated by $1,c_{1},\cdots, c_{k}$
with relations
 $$
 \{c_i,c_j\}=2\delta_{ij}. 
 $$
In the case $k=2m$,  there is an isomorphism $Cl_{2m}\cong M_{2^m}(\mathbb{C})$ as algebras,  and $Cl_{k}$ has a unique irreducible representation. 
In the case $k=2m+1$,  there is an isomorphism $Cl_{2m+1}\cong M_{2^m}(\mathbb{C})\oplus M_{2^m}(\mathbb{C})$ as algebras, and  $Cl_{k}$ has two irreducible representations. 
This result is presented, for example, in Lawson and Michelsohn \cite{Lawson} Chapter 1. 
When $k=2m+1$, $Cl_{k}$ has two   irreducible  representations  $S_{+}$ and $S_{-}$. 
 For $(\sqrt{-1})^{m}c_1\cdots c_{2m+1}$, $S_{+}$ is  the $+1$ eigenspace  and $S_{-}$ is the $-1$ eigenspace. 
Then, we select $S_{+}$ as an   irreducible  representation of $Cl_{k}$. 
In this way, we fix an irreducible representation of $Cl_k$, which is denoted by $\widetilde{M}_k$. 
To investigate eigenvalue problems of Hamiltonians, it is convenient to introduce the following  creation and  annihilation operators. 

In the case  $k=2m$, we define  the creation operators $a_i^{\dagger}$ for $1 \leq i \leq m$ and the annihilation operators $a_i$ for $1 \leq i \leq m$  by  $c_1, c_2, \cdots, c_{2m}$ as 
\begin{align*}
    \begin{cases}
        a_{i}=\frac{1}{2}(c_{2i-1}+\sqrt{-1}c_{2i})\\
        a_{i}^{\dagger}=\frac{1}{2}(c_{2i-1}-\sqrt{-1}c_{2i}). 
    \end{cases}
\end{align*}

In the case $k=2m+1$, we define  the creation  operators $a_i^{\dagger}$ for $1 \leq i \leq m$, and the annihilation operators $a_i$ for $1 \leq i \leq m$ and $b$ by  $c_1$, $c_2, \cdots, c_{2m+1}$ as 
\begin{align*}
    \begin{cases}
        a_{i}=\frac{1}{2}(c_{2i-1}+\sqrt{-1}c_{2i})\\
        a_{i}^{\dagger}=\frac{1}{2}(c_{2i-1}-\sqrt{-1}c_{2i})\\
        b=c_{2m+1}. 
    \end{cases}
\end{align*}
 The operators $a_i, a_i^{\dagger}, b$ satisfy the relations 
 \begin{align}
        \label{eq:anticom2}
        \{a_i,a_j\}=\{a_i^{\dagger},a_j^{\dagger}\}=\{a_j,b\}=\{a_j^{\dagger},b\}=0,\ \{a_i,a_j^{\dagger}\}=\delta_{ij},\ \{b,b\}=2. 
 \end{align}
   There exists a non zero vector  $|vac\rangle$ in $\widetilde{M}_k$ such that  $a_i|vac\rangle=0$ for $1\leq i\leq \lfloor\frac{k}{2}\rfloor$ and $b|vac\rangle=|vac\rangle$. 
   The vector space $\widetilde{M}_k$ is  spanned by  $|vac\rangle$ and $a_{l_1}^{\dagger}\cdots a_{l_s}^{\dagger}|vac\rangle$ for $1\leq l_1<\cdots<l_s\leq \lfloor\frac{k}{2}\rfloor$. 
        The action of $a_i$, $a_{i}^{\dagger}$ and $b$ on $a_{l_1}^{\dagger}\cdots a_{l_s}^{\dagger}|vac\rangle$ is defined in such a way that it is  compatible with the relations (\ref{eq:anticom2}). 
    When $k=2m+1$, the element $|vac\rangle$ of $S_{+}$ satisfies $b|vac\rangle=|vac\rangle$. 
    The operators $c_1,\cdots, c_k$ acting on  $\widetilde{M}_k$ are called   Majorana operators. 

   We consider the diamond crystal $\Delta_d$. 
    To each vertex $v\in V(\Delta_d)$ we associate $\widetilde{M}_v$ which is isomorphic to $\widetilde{M}_{d+2}$. 
    Then, we define the space as
    \begin{align*}
        &\bigotimes_{v\in V(\Delta_d)}\widetilde{M}_v. 
    \end{align*}

There is an action of $\Gamma_{A_d}$ on $\Delta_d$ and the quotient space $\Delta_d/\Gamma_{A_d}$ is considered as a graph. 
We call this graph the base graph of $\Delta_d$ and denote it by $X_0$. 
We have a maximal abelian covering 
\begin{align*}
    &\pi:\Delta_d\rightarrow X_0
\end{align*}
(see T.~Sunada \cite{Sunada} 8.3 example (ii)). 
The graph $X_0$ is shown in Figure \ref{base4}. 
We set the vectors 
\begin{align}\label{eq:beta}
\beta_0=-p,\quad \beta_i=\alpha_i-p \mbox{ for } 1\leq i\leq d
\end{align}
with the $\mathbb{Z}$-basis $\{\alpha_i\}_{1\leq i\leq d}$ of $A_d$,  and $p=\frac{1}{d+1}\sum_{i=1}^{d}\alpha_i$. 
We denote by $E(X_0)$  the set of the edges of $X_0$. 
We denote by $e_{i}$ the edges of $X_0$ for $0\leq i\leq d$. 
For $a \in A_d+p$ the set $\pi^{-1}(e_i)$ consists of the edges connecting $a$ and $a+\beta_i\in A_d$ for $0\leq i\leq d$. 
The covering transformation group of $\pi:\Delta_d\rightarrow X_0$ is the lattice group $\Gamma_{A_d}$, which is in one-to-one correspondence with $H_1(X_0,\mathbb{Z})$. 

\begin{figure}[htbp]
    \begin{center}
    \includegraphics[width=80mm]{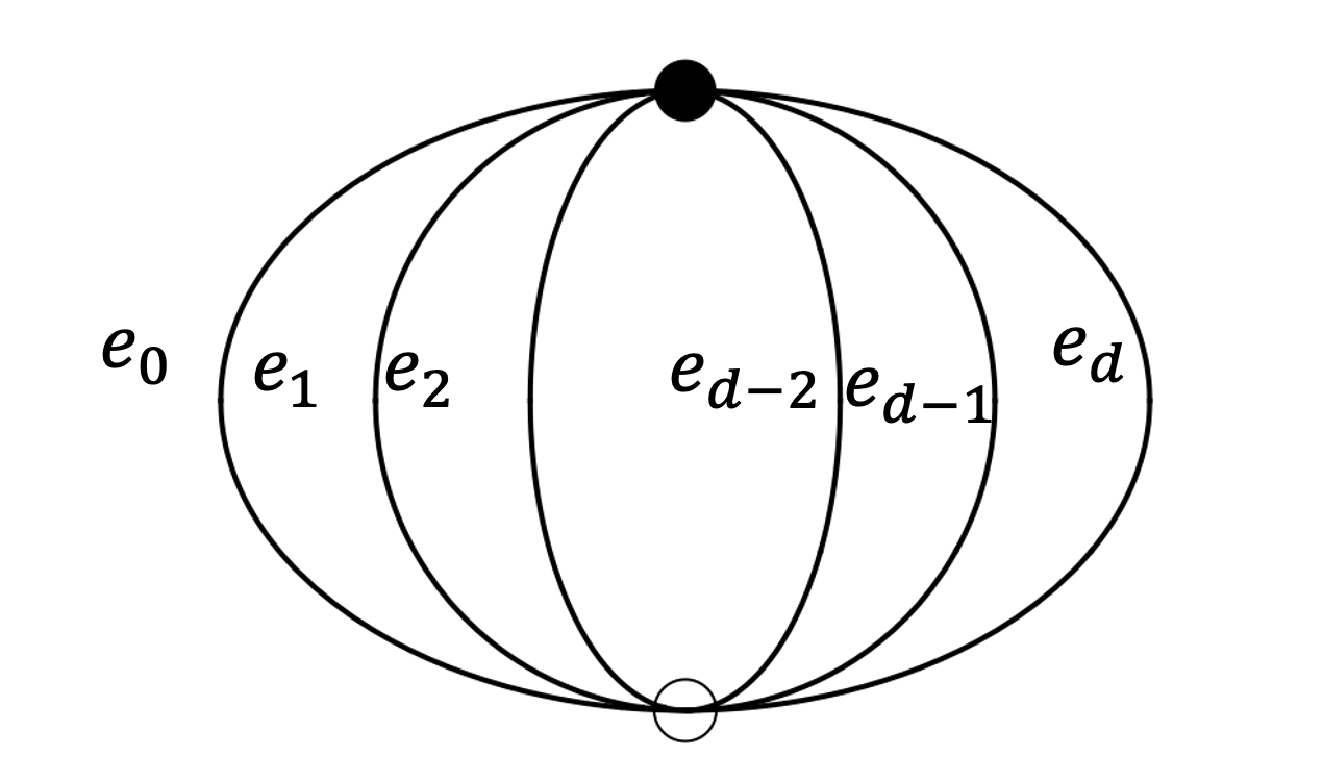}
    \caption{The base graph of $d$-dimensional diamond crystal}
    \label{base4}
    \end{center}
    \end{figure}
 
For each $i$, $0\leq i \leq d$, we choose a fundamental domain $D_{\beta_i}$ of $\Gamma_{A_d}$  as 
$$
    D_{\beta_i}=\{-\sum_{j\neq i}^{d}t_j (\beta_i-\beta_j) \mid\ 0 \leq t_j\leq 1\}. 
$$
We set $\gamma_i=\frac{d}{2}\beta_i$. 
We denote by $D'_{\beta_i}$ the shifted fundamental domain $D_{\beta_i}-\gamma_i$  as shown in Figure \ref{fundamental}. 
We set $P_1=-\frac{d}{2}\beta_i$, $P_2=-\frac{d}{2}\beta_i-\beta_i$. 

\begin{figure}[htbp]
    \begin{center}
    \includegraphics[width=100mm]{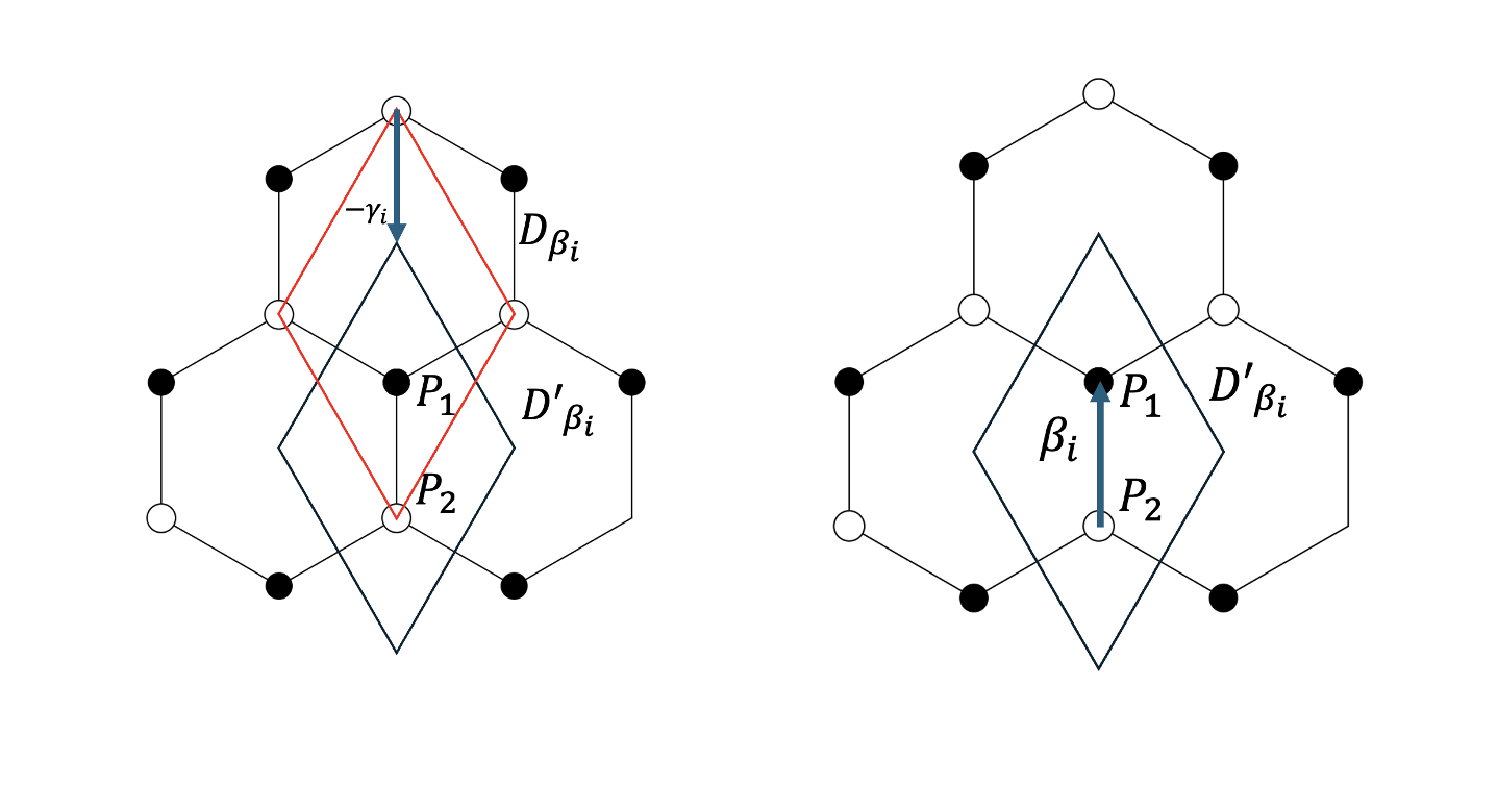}
    \caption{The shifted fundamental domain $D_{\beta_i}'$ of the honeycomb lattice}
    \label{fundamental}
    \end{center}
    \end{figure}

\begin{lem}\label{Pj}
The points $P_1$ and $P_2$ belong to the interior of $D'_{\beta_i}$ and there are no other vertices of $V(\Delta_d)$ belonging to  $D'_{\beta_i}$. 
\end{lem}
\begin{proof}
The point $P_1$ is expressed by $-\sum_{j\neq i}^{d}t (\beta_i-\beta_j)$ with $t=\frac{d}{2(d+1)}\in (0,1)$. 
The point $P_2$ is expressed by $-\sum_{j\neq i}^{d}t (\beta_i-\beta_j)$ with $t=\frac{d+2}{2(d+1)}\in (0,1)$. 
Thus, the points $P_1$ and $P_2$ belong to the interior of $D'_{\beta_i}$.
With respect to the action of $\Gamma_{A_d}$, the set of vertices $V(\Delta_d)$ is expressed as the disjoint union of two orbits
$$
V(\Delta_d)=(\Gamma_{A_d}\cdot P_1)\sqcup  (\Gamma_{A_d}\cdot P_2). 
$$
For $\Gamma_{A_d}\ni g\neq e$ we have $g\cdot P_1\notin D'_{\beta_i}$, $g\cdot P_2 \notin D'_{\beta_i}$. 
Since $D_{\beta_i}'$ is a fundamental domain both $P_1$ and $P_2$ belong to the interior of $D'_{\beta_i}$. 
Thus the other vertices of $V(\Delta_d)$ do not belong to $D'_{\beta_i}$. 
\end{proof}

We define the labeling the edges of $X_0$ as $\ell: E(X_0) \rightarrow \mathbb{Z}$ where $\ell(e_{i})=i+1, 0\leq i\leq d$. 
When $(v,v')\in E(\Delta_d)$, we call $\pi((v,v'))$ the spin direction of the edge $(v,v')$ and $\ell(\pi((v,v')))$ the labeling of the edge $(v,v')$. 
We associate the Majorana operators to the vertices and the edges of the base graph as shown in Figure \ref{Majorana2}. 

\begin{figure}[htbp]
\begin{center}
\includegraphics[width=70mm]{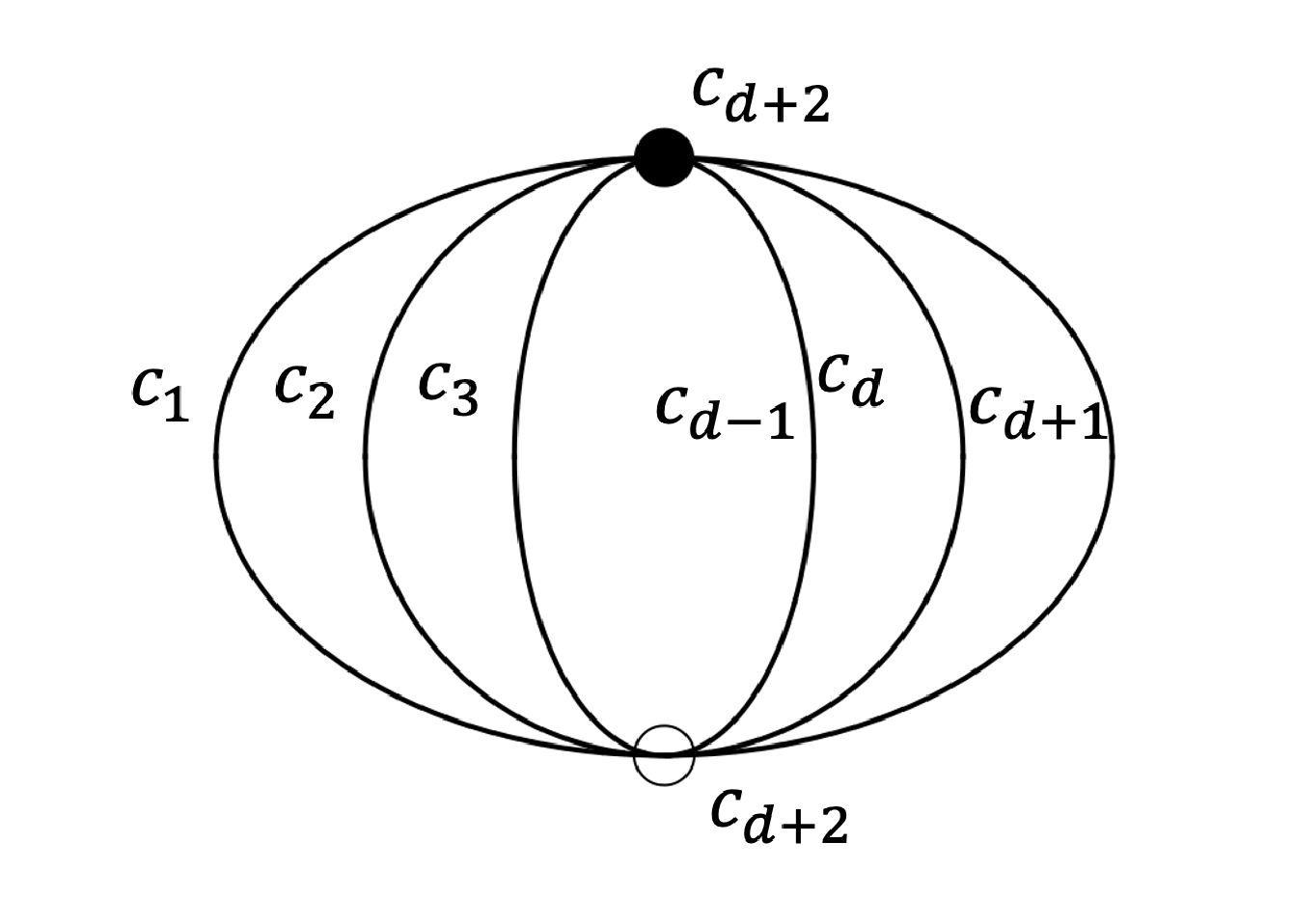}
\caption{The base graph of the $d$-dimensional diamond crystal and the Majorana operators}
\label{Majorana2}
\end{center}
\end{figure}

\section{The Hamiltonian of of the Kitaev model for $\Delta_d$}
In this section, we define the Hamiltonian of the Kitaev model for $\Delta_d$. 
We set $\sigma^{k}=\sqrt{-1}c_{k}c_{d+2}$. 
To each vertex $v$ we associate the operator $c_k^v$, which is the action of $c_k$ on $\widetilde{M}_v$ and  $Id$ on the other components of $\bigotimes_{v\in V(\Delta_d)}\widetilde{M}_v$. 
We set $\sigma_v^{k}=\sqrt{-1}c^v_{k}c^v_{d+2}$. 
\begin{dfn}
    We define the Hamiltonian as
    \begin{align}\label{eq:ham}
         H=-\sum_{e\in E(X_0)}\sum_{(v,v')\in \pi^{-1}(e)}J_{\ell(e)}\sigma_{v}^{\ell(e)}\sigma_{v'}^{\ell(e)}
    \end{align}
    where $J_{\ell(e)}\in \mathbb{R}$. 
\end{dfn}
We set 
$$\hat{u}_{v,v'}=\sqrt{-1}c_{\alpha_{v,v'}}^{v}c_{\alpha_{v,v'}}^{v'}$$
where $c_{\alpha_{v,v'}}^{v}$ is the Majorana operator, and $\alpha_{v,v'}\in \ell(E(X_0))$ is the the labeling of the edge $(v,v')$. 
The operator $H$ is also expressed as
$$H =\frac{\sqrt{-1}}{4}\sum_{v,v'\in V(\Delta_d)} \hat{A}_{v,v'}c_{v}c_{v'}$$
with
$$
\hat{A}_{v,v'}=
\begin{cases}
    2J_{\ell(e)}\hat{u}_{v,v'},&\quad \mbox{$(v,v') \in \pi^{-1}(e) $} \\
    0,&\quad \mbox{otherwise}
\end{cases}
$$
where $c_v$ and $c_{v'}$ are the Majorana operators $c_{d+2}^{v}$ and $c_{d+2}^{v'}$ for $v,v' \in V(\Delta_d)$. 
    We define the operator $D$ acting on the space $\widetilde{M_v}$ as
    $$D=(-1)^{\lfloor\frac{d+1}{2}\rfloor}(\frac{1}{\sqrt{-1}})^{\lfloor \frac{d}{2}\rfloor+1}\prod_{i-1}^{d+1}c_ic_{d+2}.$$ 
   We set 
    $$M_v=\{v\in \widetilde{M_v}\mid Dv=v\}.$$ 
The operator $D$ is also described as 
    $$D=(-1)^{\lfloor\frac{d}{2}\rfloor+1}\prod_{i=1}^{\lfloor\frac{d}{2}\rfloor+1}(1-2a_i^{\dagger}a_i).$$ 
The spectra of $D$ are $1$ and $-1$ since $D^2=Id$. 
Let $\nu(d)$ be the dimension of the space $M_v$. 
We have $\nu(d)=2^{\lfloor\frac{d}{2}\rfloor}$. 
We define the action of $\widetilde{D}$ on the space $\bigotimes_{v\in V(\Delta_d)}\widetilde{M_v}$ as 
$$\widetilde{D}(\bigotimes_{v\in V(\Delta_d)} u_v)=(\bigotimes_{v\in V(\Delta_d)} Du_v).$$ 
We set 
    $$
       M(\Delta_d)=\{u \in \bigotimes_{v\in V(\Delta_d)}\widetilde{M_v} \mid \widetilde{D}u=u\}. 
    $$
   We observe that the operators $H$ and $\widetilde{D}$ commute. 
    Thus, $M(\Delta_d)$ is invariant by the action of $H$.

\section{ Fourier transform}
In this section, we describe spectra of the Hamiltonian of the Kitaev model for $\Delta_d$ by the Fourier transform. 
\begin{lem}
    For adjacent vertices $v,v' \in \Delta_d$, the eigenvalues of the operator $\hat{u}_{v,v'}$ are $1$ and $-1$. 
    The space $M(\Delta_d)$ is decomposed into the eigenspaces of the eigenvalues $1$ and $-1$ as $M_{+}(\Delta_d)_{v,v'}\oplus M_{-}(\Delta_d)_{v,v'}$. 
\end{lem}
\begin{proof}
    The operator $\hat{u}_{v,v'}=\sqrt{-1}c_{\alpha_{v,v'}}^{v}c_{\alpha_{v,v'}}^{v'}$ acts on the space $\bigotimes_{v\in V(\Delta_d)}\widetilde{M}_v$ and $M(\Delta_d)$ since $\hat{u}_{v,v'}$ commutes with $\widetilde{D}$. 
     The operator $c_{\alpha_{v,v'}}^{v}$ satisfies $(c_{\alpha_{v,v'}}^{v})^2=1$. 
     Therefore the eigenvalues of $\hat{u}_{v,v'}$ are $\pm1$, and we have a direct sum decomposition  into eigenspaces $M_{+}(\Delta_d)_{v,v'}\oplus M_{-}(\Delta_d)_{v,v'}$. 
\end{proof}
     We set $$M_{+}(\Delta_d)=\bigcap_{v, v' \,\mbox{adjacent}}M_{+}(\Delta_d)_{v,v'}.$$ 

As in (2.1), $\{\alpha_i\}$ denotes the $\mathbb{Z}$-basis of $A_d$. 
  The Hamiltonian $H$ acting on $M_{+}(\Delta_d)$ is expressed as 
    $$
        H =\frac{\sqrt{-1}}{4}\sum_{v,v'\in V(\Delta_d)} A_{v,v'}c_{v}c_{v'}
    $$
    with
    $$
        A_{v,v'}=
        \begin{cases}
            2J_{\ell(e)}, &\quad \mbox{$(v,v') \in \pi^{-1}(e) $}\\
                        0,&\quad \mbox{otherwise}. 
        \end{cases}
    $$
We represent $v$ as $s\lambda$. 
Here $s=1$ if $v$ belongs to $\Gamma_{A_d}\cdot P_1$ and $s=0$ if $v$ belongs to $\Gamma_{A_d}\cdot P_2$. 
The symbol $\lambda\in \Gamma_{A_d}$ shows that $v\in \lambda\cdot D'_{\beta_i}$. 
In the case where $(v,v')\in \pi^{-1}(e)$, the vertices $v$ and $v'$ belong to interior of $D'_{\beta_0}, \cdots, D'_{\beta_d}$. 
We write the vertices of $V(X_0)$ as $0, 1$. 
Hence, the map $\pi$ can be written as $\pi(v)=s$. 
The fundamental domain $D'_{\beta_j}$ of the honeycomb lattice for $0\leq j\leq 2$ is shown in Figure \ref{fundamental2}. 
\begin{figure}[htbp]
    \begin{center}
    \includegraphics[width=60mm]{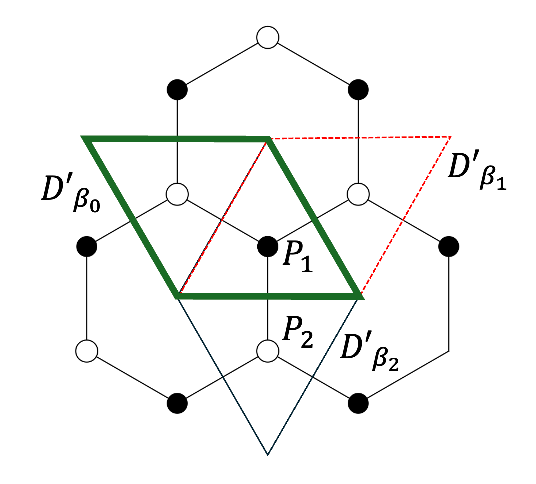}
    \caption{The fundamental domains of the honeycomb lattice $D'_{\beta_0},D'_{\beta_1}$, and $D'_{\beta_2}$}
    \label{fundamental2}
    \end{center}
    \end{figure}
Thus, the vertex $v$ is contained in $D'_{\beta_j}$ for $j \in\{0,\cdots,d\}$. 
Then we also describe $H$ as 
    $$
        H =\frac{\sqrt{-1}}{4}\sum_{(s\lambda),(t\mu) \in V(\Delta_d)} A_{s\lambda,t\mu}c_{s\lambda}c_{t\mu}
    $$
    with
    $$
        A_{s\lambda,t\mu}=
    \begin{cases}
        2J_{\ell(e)}, &\quad \mbox{$(s\lambda,t\mu) \in \pi^{-1}(e) $} \\
       0,&\quad \mbox{otherwise}. 
    \end{cases}
    $$
 
     
    In general, let $X$ be a topological crystal with a base graph $Y$ in the sense of Sunada \cite{Sunada}. 
    In this situation, Kato and Richard \cite{Kato}  establish a unitary transformation  $I\mathcal{F}\mathcal{U}_U:l^2(X,\mathbb{C})\to L^2(T^d,\mathbb{C}^{|Y|})$ where
     $$
     \mathcal{F}:l^2(\mathbb{Z}^d, l^2(Y))\to L^2(T^d,l^2(Y))
     $$
      is the Fourier transform and $|Y|$ is the number of vertices of $Y$. 
     We adapt this method to the diamond crystal $\Delta_d$. 
     We fix the fundamental domain $D'_{\beta_i}$. 
     As in Lemma \ref{Pj}, we take $P_1, P_2 \in D'_{\beta_i}$. 
     Let us recall that $V(\Delta_d)$ is decomposed into $2$ orbits as
     $$
     V(\Delta_d)=(\Gamma_{A_d}\cdot P_1)\sqcup  (\Gamma_{A_d}\cdot P_2). 
     $$
     There is a bijection $\psi: V(\Delta_d)\to A_d \times V(X_0)$ defined by 
     $$
     \psi (\mu\cdot P_1) =(\mu, 0)
    $$ 
    $$
    \psi (\mu\cdot P_2) =(\mu, 1). 
    $$
 The space $l^2(\Delta_d,\mathbb{C}^{\nu(d)})$ is defined as  
 $$
 l^2(\Delta_d,\mathbb{C}^{\nu(d)}):=\{f : V(\Delta_d) \to \mathbb{C}^{\nu(d)}\mid \sum_{v\in V(\Delta_d)}\sum_{i=1}^{\nu(d)}|f_i(v)|^2<\infty\}. 
 $$
  We fix an identification 
  $$
  l^2(X_0) \otimes \mathbb{C}^{\nu(d)} \cong \mathbb{C}^{\nu(d)}\oplus\mathbb{C}^{\nu(d)}
  $$
   so that the first component corresponds to $\Gamma_{A_d}\cdot P_1$, and the second component corresponds to $\Gamma_{A_d}\cdot P_1$. 

      The unitary transformation  
      $$
      U:l^2(\Delta_d, \mathbb{C}^{\nu(d)}) \to l^2(A_d,l^2(X_0)\otimes\mathbb{C}^{\nu(d)})
      $$  
      is defined by 
     \begin{equation}
     \label{U}
     [Uf](\mu)= (f(\mu\cdot P_1), f(\mu\cdot P_2)). 
     \end{equation}
    We write $f(\mu\cdot P_1)$ as $f_0(\mu)$, and $f(\mu\cdot P_2)$ as $f_1(\mu)$.


     
     As in (\ref{eq:root}), we consider $A_d$ as a lattice in $W$. 
     Let $W^*$ denote the dual space of $W$. 
     We set $\bm{x}=\sum_{i=1}^dx_i\alpha_i\in W$ and $\bm{q}=\sum_{i=1}^d\xi_ib_i\in W^*$ with $\langle b_j,\alpha_i\rangle=\delta_{ij}$. 
     Here $\langle \, , \rangle:W^*\times W\to \mathbb{R}$ is the canonical pairing. 
     We write $\langle\bm{q},\bm{x}\rangle$  as $\bm{q}\cdot \bm{x} = \sum^d_{j=1} \xi_j x_j$. 
      Let $\check{D}$ be a fundamental domain of the reciprocal lattice of $A_d$ in $W^*$. 
      We consider the Fourier transform 
          $$
          \mathcal{F}:l^2(A_d,l^2(X_0)\otimes\mathbb{C}^{\nu(d)}) \to L^2(\check{D},l^2(X_0)\otimes\mathbb{C}^{\nu(d)})
          $$
           defined by
           \begin{align*}
     [\mathcal{F}h](\bm{q})=\sum_{\mu\in A_d}e^{-\sqrt{-1}\bm{q}\cdot \mu }h(\mu). 
     \end{align*}
     for $h\in  l^2(A_d,l^2(X_0)\otimes\mathbb{C}^{\nu(d)})$ and $\xi \in\check{D}$. 
     By using a similar method as in \cite{Kato}, we obtain the inverse Fourier transform defined by
      \begin{align*}
      [\mathcal{F}^*u](\mu)=\int_{\check{D}}e^{2\pi\sqrt{-1}\xi\cdot \mu}u(\xi)d\xi
       \end{align*}
       for $u \in C(\check{D},l^2(X_0)\otimes\mathbb{C}^{\nu(d)})$ and $\mu \in A_d$ with $d\xi$  scaled so that $\int_{\check{D}}1d\xi=1$. 

      We define  the map 
      $$
      I:L^2(\check{D},l^2(X_0)\otimes\mathbb{C}^{\nu(d)})\to L^2(\check{D},\mathbb{C}^2\otimes\mathbb{C}^{\nu(d)})
      $$ 
      by using the identification (\ref{U}). 
     We write
     $$
     [Ig](\xi)=(g_0(\xi),g_1(\xi))
     $$
     for $g\in L^2(\check{D},l^2(X_0)\otimes\mathbb{C}^{\nu(d)})$ and $\xi \in \check{D}$. 
     
      The transformation $I\mathcal{F}U$  is unitary since $U$, $\mathcal{F}$, $I$ are  unitary. 
      It is expressed by the following diagram
$$
l^2(\Delta_d, \mathbb{C}^{\nu(d)})\rightarrow l^2(A_d,l^2(X_0)\otimes\mathbb{C}^{\nu(d)})\rightarrow L^2(\check{D},l^2(X_0)\otimes\mathbb{C}^{\nu(d)})\rightarrow L^2(\check{D},\mathbb{C}^2\otimes\mathbb{C}^{\nu(d)}). 
$$

 A minimum ground state energy function $\xi(\bm{q})$ is defined as follows. 
\begin{enumerate}
\item $\xi(\bm{q})$ is an eigenvalue of the operator $H$ acting on the space $M_{+}(\Delta_d)$. 
\item  An eigenfunction $f$ of $\xi(\bm{q})$ belongs to the space $l^2(\Delta_d,\mathbb{C}^{\nu(d)})$ and satisfies  the Bloch condition $f(t_{\alpha_i}\cdot v)=e^{\sqrt{-1}\bm{q}\cdot \alpha_i}f(v)$. 
\end{enumerate}
 This definition is motivated by a physical argument due to E.~H.~Lieb \cite{Lieb}. 


\begin{thm}
    \label{thm:spec}
    The minimum ground state energy functions of the Kitaev model of the $d$-dimensional diamond crystal  $\Delta_d$ are expressed as 
    $$
        \xi(\bm{q})=\pm2|J_1+\sum_{i=1}^{d}J_{i+1}e^{\sqrt{-1}\bm{q}\cdot \alpha_{i}}|. 
    $$ 
\end{thm}
\begin{proof}
       Applying the unitary transformation $I\mathcal{F}U:l^2(\Delta_d,\mathbb{C}^{\nu(d)})\to L^2(\check{D},\mathbb{C}^2\otimes\mathbb{C}^{\nu(d)})$  to the eigenfunction $f$, we obtain
      \begin{align*}
      [I\mathcal{F}Uf](\bm{q})=\tilde{f}(\bm{q})=
     \begin{pmatrix}
        \sum_{\lambda\in \Gamma_{A_d}}e^{-\sqrt{-1}\bm{q}\cdot \mu_{\lambda}}f(0\lambda )\\
        \sum_{\lambda\in \Gamma_{A_d}}e^{-\sqrt{1}\bm{q}\cdot  \mu_{\lambda}}f(1\lambda)
       \end{pmatrix}
        =
       \begin{pmatrix}
      \tilde{f}_0(\bm{q})\\
      \tilde{f}_1(\bm{q})
       \end{pmatrix}
      \end{align*}
      where $0\lambda, 1\lambda$ are  $s\lambda=v\in V(\Delta_d)$ and  $\mu_{\lambda}\in A_d$ corresponding to translation $\lambda\in \Gamma_{A_d}$. 
      
    We consider the Hamiltonian $(I\mathcal{F}U)H(I\mathcal{F}U)^{-1}$, which is regarded as a function of $\bm{q}$. 
    We write   $(I\mathcal{F}U)H(I\mathcal{F}U)^{-1}$ as $H(\bm{q})$. 
    The Hamiltonian $H(\bm{q})$ is also written as 
    $$
        H(\bm{q}) =\frac{1}{2}\sum_{\lambda,\mu} \tilde{A}_{\lambda,\mu}(\bm{q})a_{-\bm{q}\lambda}a_{\bm{q}\mu}
    $$
    with
    $$
        \tilde{A}_{\lambda,\mu}(\bm{q})=\sum_{s}A_{0\lambda,s\mu}e^{\sqrt{-1}\bm{q}\cdot \bm{r}_s},\ \
        a_{\bm{q}\lambda}=\sum_{s}c_{s\lambda}e^{-\sqrt{-1}\bm{q}\cdot \bm{r}_s},
     $$
    where $\bm{r}_s$ is the vector from $0\lambda$ to $s\lambda$ within the fundamental domain $\lambda\cdot D'_{\beta_j}$. 
    The operators $a_{\bm{q},\lambda}$ and $a^{\dagger}_{\bm{q},\mu}$ satisfy the relations
    \begin{align*}
        &a^{\dagger}_{\bm{q}\lambda}=a_{-\bm{q}\lambda}\\
        &\{a_{\bm{q}\lambda},a^{\dagger}_{\bm{q}\mu}\}=\delta_{\lambda\mu}. 
    \end{align*}
The Hamiltonian $H(\bm{q})$ is written as 
\begin{align*}
    H(\bm{q})&=\frac{\sqrt{-1}}{2}\sum_{\lambda}(\sum_{i=0}^{d}J_{i+1}e^{\sqrt{-1}\bm{q}\cdot \beta_i})a_{\bm{-q},\lambda}a_{\bm{q},\lambda}
\end{align*}
where $\sum_{i=0}^{d}\beta_i$=0. 
We set  $f(\bm{q})=2\sum_{i=0}^{d}J_ie^{\sqrt{-1}\bm{q}\cdot \beta_i}$. 
The action of $H(\bm{q})$ on $M_{+}(\Delta_d)$ is expressed as
\begin{align}
    \label{eq:hamiltonian1}
    H(\bm{q})=\frac{1}{4}\sum_{\lambda}(a_{\bm{-q},\lambda}a_{\bm{q},\lambda})
    \begin{pmatrix}
        O&\sqrt{-1}f(\bm{q})\\
        -\sqrt{-1}f(\bm{q})^*&O
    \end{pmatrix}
    \begin{pmatrix}
        a_{\bm{-q},\lambda}\\
        a_{\bm{q},\lambda}
    \end{pmatrix}. 
\end{align}
This expression represents a Hamiltonian $H(\bm{q})$ involving two Majorana operators. 
We set
\begin{align*}
    \sqrt{-1}\tilde{A}(\bm{q})=
    \begin{pmatrix}
        O&\sqrt{-1}f(\bm{q})\\
        -\sqrt{-1}f(\bm{q})^*&O
    \end{pmatrix}
\end{align*}
by (\ref {eq:hamiltonian1}). 
The eigenvalues of the matrix $\sqrt{-1}\tilde{A}(\bm{q})$  are $\pm |f(\bm{q})|$. 
We set  $\xi(\bm{q})=\pm|f(\bm{q})|$. 
Thus, we compute the eigenvalues as
\begin{align*}
    &\xi(\bm{q})=\pm2|\sum_{i=0}^{d}J_{i+1}e^{\sqrt{-1}\bm{q}\cdot \beta_i}|\\
    &=\pm2|e^{\sqrt{-1}\bm{q}\cdot \beta_0}||J_1+\sum_{i=1}^{d}J_{i+1}e^{\sqrt{-1}\bm{q}\cdot \alpha_i}|\\
    &=\pm2|J_1+\sum_{i=1}^{d}J_{i+1}e^{\sqrt{-1}\bm{q}\cdot \alpha_i}|. 
\end{align*}

This completes the proof. 
\end{proof}

\section{Zeros of the energy functions and  energy gaps.}
In this section, for the energy function $\xi(\bm{q})$, we describe  zeros  and energy gaps. 
The corresponding results in the case $d=2$ are due to  A.Kitaev \cite{Kitaev}. 

We consider $(J_0,\cdots,J_d)$ as the parameters in the equation (\ref{eq:ham}). 
\begin{thm}\label{lem:gap}
    For $J_i\in \mathbb{R}$, $0\leq i\leq d$, the inequalities  
    \begin{align}\label{eq:one}
        |J_i|\leq\sum_{0\leq j\leq d,i\neq j}|J_j|\, \mbox{ for all }i, 0\leq i\leq d
    \end{align}
are satisfied if and only if  there exists $\bm{q}\in\mathbb{R}^{d+1}$ such that $\xi(\bm{q})=0$. 
\end{thm}
The following lemma might be a well-known fact, although we provide a proof since we could not find it in the literature. 
\begin{lem}\label{lem:polygon}
    We suppose $0< a_0\leq \cdots\leq a_d$. The inequality
    \begin{align}\label{eq:two}
        a_d<\sum_{j=0}^{d-1}a_j
    \end{align}
    is satisfied if and only if there exists a $(d+1)$-gon such that the lengths of the sides are $a_0, a_1,\cdots, a_d$. 
\end{lem}
\begin{proof}
   We suppose that there exists a $(d+1)$-gon such that the lengths of the sides are $a_0, a_1,\cdots, a_d$. 
    Since a side is the shortest length connecting  two endpoints of a edge of a polygon, the inequalities (\ref{eq:two}) hold. 

    Conversely, we suppose the inequalities  (\ref{eq:two}). 
    We prove the statement by induction on $d$. 

    First, we consider the case $d=2$.
    Then the statement holds because of the triangle inequality. 

    Next, we assume that the statement holds in the case  $d-1$. 
    We choose $\epsilon>0$  such that $\epsilon<a_0$ and $\epsilon<\sum_{j=0}^{d-1}a_j-a_d$. 
We set $e=a_d-a_0+\epsilon$. 
Since the inequalities 
 $$a_0<e+a_d,\quad a_d<e+a_0=a_d+\epsilon,\quad e=a_d-(a_0-\epsilon)<a_0+a_d$$
 hold, there exists a triangle such that the lengths of the sides are $e$, $a_0$, $a_d$.

We consider the following cases (1) and (2). 

\noindent
    (1) In the case $e>a_{d-1}$, the inequality 
    $$
    e<\sum_{i=1}^{d-1}a_i
    $$
    holds. 

\noindent
    (2) In the case $e\leq a_{d-1}$, the inequality 
    $$
    a_{d-1}<e+\sum_{i=1}^{d-2}a_i=a_d+\epsilon+a_1-a_0+\sum_{i=2}^{d-2}a_i
    $$
    holds.

In both cases, by hypothesis of induction there exists a $d$-gon such that the lengths of the sides are $e, a_1,\cdots,a_{d-1}$. 
We attach the $d$-gon and the triangle by identifying them along the side of  the length  $e$. 
This construction yields a $(d+1)$-gon. 
By choosing $\epsilon$ sufficiently small, the two polygons sharing the side of length  $e$ can be arranged so that they do not overlap. 
Therefore there exists a $(d+1)$-gon such that the lengths of the sides are $a_0, a_1,\cdots, a_d$. 
\end{proof}

We prove Theorem \ref{lem:gap}. 
\begin{proof}
    We suppose that the inequalities (\ref{eq:one}) are satisfied. 
    If the inequality $|J_i|<\sum_{i\neq j,0\leq j\leq d}|J_j|$ holds for any $i, 0\leq i\leq d$, then by Lemma \ref{lem:polygon} there exists a $(d+1)$-gon such that the lengths of the sides are $|J_0|, |J_1|,\cdots, |J_d|$. 
    Thus, there exist $\theta_0,\cdots, \theta_d$ such that
    \begin{align}\label{eq:angle}
        \sum_{i=0}^{d}J_ie^{\sqrt{-1}\theta_i}=0. 
    \end{align} 
    If there exists $i$, $0\leq i\leq d$ such that $|J_i|=\sum_{j=1,i\neq j}^{d}|J_j|$, then we have $\theta_0,\cdots, \theta_d$ such that the equation (\ref{eq:angle}) holds since
    $$\sum_{i\neq j,0\leq j\leq d}|J_j|-|J_i|=0.$$ 

    For $\beta_0,\cdots,\beta_d$ in the equation (\ref{eq:beta}), we consider a system of linear equations 
    \begin{align}\label{eq:linear}
    \bm{q}\cdot\beta_i=\theta_i \mbox{ for } i, 0\leq i\leq d
     \end{align} for $\bm{q}\in\mathbb{R}^{d+1}$. 
     Since the vectors $\beta_0,\cdots,\beta_d$ are linearly independent, the system of equations (\ref{eq:linear}) has a unique solution. 
     For such $\bm{q}$, we have $\xi(\bm{q})=0$. 
      Therefore there exists $\bm{q}\in\mathbb{R}^{d+1}$ such that $\xi(\bm{q})=0$. 
    
    Conversely, we suppose that there exists $\bm{q}\in\mathbb{R}^{d+1}$ such that $\xi(\bm{q})=0$. 
    By Lemma \ref{lem:polygon}, if there exists $i$, $0\leq i\leq d$ such that $|J_i|>\sum_{0\leq j\leq d,j\neq i}|J_j|$, then
    $$\sum_{i=0}^dJ_ie^{\sqrt{-1}\theta_i}\neq0$$
    for all $\theta_i\in\mathbb{R}$. 
        Therefore the inequalities (\ref{eq:one}) are satisfied. 
    This completes the proof of  Theorem \ref{lem:gap}. 
\end{proof}

We define the simplex $\Phi_d$ as
$$\Phi_{d}=\{(x_0,\cdots,x_d)\in \mathbb{R}^{d+1}\mid\sum_{i=0}^dx_i=1\mbox{ and }x_i\geq0\}.$$ 
We define the region $\Omega_d$ as
$$\Omega_{d}=\{(x_0,\cdots,x_d)\in \Phi_d\mid x_i > \frac{1}{2}\mbox{ for some }i, 0\leq i\leq d \}.$$ 
We show that energy gaps appear for $(|J_0|,\cdots,|J_d|)\in \Omega_d$. 
\begin{thm}\label{thm:gap2}
    We suppose that $(|J_0|,\cdots |J_d|)\in \Phi_d$. Then for any $\bm{q}\in\mathbb{R}^{d+1}$, we have $\xi(\bm{q})\neq0$  if and only if the condition $(|J_0|, \cdots, |J_d| )\in \Omega_d$ holds. 
\end{thm}
\begin{proof}
We suppose that the condition $(|J_0|, \cdots, |J_d| )\in \Omega_d$ holds. 
In the simplex $\Phi_d$, by Theorem \ref{lem:gap}, if there exists some $i$, $0\leq i\leq d$ such that $|J_i|>\sum_{0\leq j\leq d, j\neq i}|J_j|$, then $\xi(\bm{q})\neq 0$ for all $\bm{q}\in\mathbb{R}^{d+1}$. 
By $\sum_{i=0}^d|J_i|=1$, if we have $|J_i|>\frac{1}{2}$, then  $\sum_{0\leq j\leq d,i\neq j}|J_j|<\frac{1}{2}$. 
Thus, in the simplex $\Phi_d$ if there exists a $i$, $0\leq i\leq d$ such that $|J_i|>\frac{1}{2}$, then the inequalities (\ref{eq:one}) are not satisfied. 
Therefore for any $\bm{q}\in\mathbb{R}^{d+1}$, we have $\xi(\bm{q})\neq0$. 

Conversely, we suppose that for any $\bm{q}\in\mathbb{R}^{d+1}$ we have $\xi(\bm{q})\neq0$. 
We assume that the condition $(|J_0|, \cdots,|J_d|)\notin \Omega_d$ holds. 
Then we have $|J_i|\leq\frac{1}{2}$ for all $i$, $0\leq i\leq d$. 
Since the inequalities (\ref{eq:one}) holds when the parameters $(|J_0|,\cdots, |J_d|)$ satisfy $\sum_{i=0}^d|J_i|=1$, there exists $\bm{q}\in \mathbb{R}^{d+1}$ such that $\xi(\bm{q})=0$. 
Therefore, the condition $(|J_0|, \cdots,|J_d| )\in \Omega_d$ holds. 
\end{proof}

Theorem \ref{thm:gap2}  shows that   energy gaps appear in the region $\Omega_d$. 

\section{Relation with the tight binding model}
In this section, we compare the minimum ground state energy functions of the Kitaev model with the energy functions of the tight binding model. 
We treat the Schr\"{o}dinger equation with a periodic potential $V(x)$ expressed as
$$
    \hat{H}\psi=E\psi
$$
with
$$
    \hat{H}=-\frac{\hbar^2}{2m}\Delta+V(x). 
$$
Let $G$ be the crystallographic group for $\Delta_d$. 
We suppose that $V(x)$ is invariant by the action of $G$. 
This Schr\"{o}dinger operator $\hat{H}$ is a linear map $L^2(\mathbb{R}^d)\to L^2(\mathbb{R}^d)$. 
As in section 4 and 5, the domain $D'_{\beta_j}$ is a fundamental domain of $A_d$ in $W$ and the domain $\check{D}$ is a fundamental domain of the reciprocal lattice of $A_d$ in $W^*$. 
Then, the operator $\hat{H}$ on $L^2(\mathbb{R}^d)$ is unitarily equivalent to the direct-integral decompositions of the Hamiltonians $\hat{H}(\bm{q})$ acting on $L^2(D'_{\beta_j})$ with respect to $\check{D}$ (see, for example in P.~Kuchment \cite{Kuchment2016}, section 4). 

As an approximation model for the Schr\"{o}dinger equation, we define the tight binding model for $\Delta_d$ by using the base graph $X_0$. 
    The space $L^2_{loc}(\mathbb{R}^d)$ denotes 
 $$
\left\{
f : \mathbb{R}^d \to \mathbb{C}
\;\middle|\;
f \in L^2(V) \text{ for every compact } V \subset \mathbb{R}^d
\right\}. 
$$
    We consider a family of functions $\psi_v\in L^2_{loc}(\mathbb{R}^d)$, $v\in V(\Delta_d)$ satisfying the following conditions. 
    \begin{enumerate}
    \item For each fixed $\bm{q}\in \check{D}$, $\langle\psi_v|\psi_{v'}\rangle=\int_{D'_{\beta_j}}\psi_v^{*}\psi_{v'}d\bm{x}=\delta_{v,v'}$ for $v,v'\in V(\Delta_d)$. 
   \item $\psi_v(x)$ is written as $e^{\sqrt{-1}\bm{q}\cdot x}u_v(x)$, $v\in V(\Delta_d)$ where $u_v(x)$ is invariant by the action of $G$ and $\bm{q}$ belongs to $\check{D}$. 
    \end{enumerate}
    Let $\psi^0_{s}(x)$ be the function defined by
    $$
    \psi^0_{s}(x)=\sum_{v\in \pi^{-1}(s)}\psi_v(x),
    $$
     with $s\in V(X_0)$. 
     Then, $\psi_{0}^0$ and $\psi_{1}^0$ are the linearly independent. 
       The vertices $v$ and $v'$ are nearest neighbors if and only if $v$ and $v'$ are adjacent. 
    We denote by $[v,v']$ the oriented edge connecting $v$ and $v'$. 
    We suppose that
    $$
       h_{v,v'}=\int_{D'_{\beta_j}}{\psi_v}^*\hat{H}\psi_{v'}dx
    $$
        is given as
    $$
        h_{v,v'}=
        \begin{cases}
            0,&\quad v=v' \\
            \sum_{[v,v']\in \pi^{-1}(e')}t_{v,v'} e^{\sqrt{-1}\bm{q}\cdot r_{v,v'}}, &\quad v\mbox{ and}\,v' \mbox{are adjacent}\\
            0, &\quad \mbox{otherwise}
        \end{cases}
    $$
    with $t_{v,v'}=t_{v',v}^*$, $v,v'\in V(\Delta_d)$, $e'\in E(X_0)$ and where $r_{v,v'}$ is the vector representing  the oriented edge $[v, v']$. 
When the vertices $v,v'\in V(\Delta_d)$ are adjacent and  the vector of $[v,v']$ is $\beta_i$, we write $t_{v,v'}$ as $t_{i+1}$ for $0\leq i \leq d$. 
Then, $\hat{H}$ which acts on $(\psi_{0}^0, \psi_{1}^0)^T$ is written as
   \begin{align*}
   \hat{H}
    \begin{pmatrix}
    \psi_{0}^0\\
    \psi_{1}^0
   \end{pmatrix}=
   \begin{pmatrix}
    0&r(\bm{q})\\
    r(\bm{q})^*&0
   \end{pmatrix}
   \begin{pmatrix}
    \psi_{0}^0\\
    \psi_{1}^0
   \end{pmatrix}
\end{align*}
where $r(\bm{q})=\sum_{i=0}^{d}t_{i+1}e^{\sqrt{-1}\bm{q}\cdot \beta_i}$. 
We set
  \begin{align*}
 A=
\begin{pmatrix}
    0&r(\bm{q})\\
    r(\bm{q})^*&0
   \end{pmatrix}. 
   \end{align*}
The eigenvalues $E(\bm{q})$ of this matrix $A$ are energy functions of the tight binding model. 
Thus, $E(\bm{q})=\pm|r(\bm{q})|=\pm|t_1+\sum_{i=1}^{d}t_{i+1}e^{\sqrt{-1}\bm{q}\cdot \alpha_i}|$. 
This result shows that
the energy functions of the tight binding model of $\Delta_d$ are 
    \begin{align}
        \label{eq:energy}
        E(\bm{q})=\pm |t_1+\sum_{i=1}^{d}t_{i+1}e^{\sqrt{-1}\bm{q}\cdot \alpha_i}|. 
    \end{align}

In \cite{2}, M. Tsuchiizu applies a similar method using the base graph for the tight binding model of the \( K_4 \) lattice. 
From Theorem \ref{thm:spec} and  (\ref{eq:energy}), when we set $2J_i=t_i$  the minimum ground state energy functions of the Kitaev model of the $d$-dimensional diamond crystal coincide with the energy functions of the tight binding model of the $d$-dimensional diamond crystal. 
We refer to \cite{Katsunori}, \cite{Keown}, and \cite{Takahashi}  for related works on the $3$-diamond crystal. 

\section{Acknowledgments}
The author would like to express his deepest gratitude to his supervisor Toshitake~Kohno for his invaluable guidance and support throughout the course of this research. 
Thanks are also due to Ayuki~Sekisaka for helpful comments and discussions.



\begin{thebibliography}{99}  
    \bibitem{Katsunori}K.~Katsunori, {Surface Structure and Surface State of a Tight Binding Model on a Diamond Lattice}, {arXiv preprint arXiv:2504.21223}, 2025.
    \bibitem{Kato} K.~Kato and S.~Richard, {Topological Crystals: Independence of Spectral Properties with Respect to Reference Systems}, Symmetry, 16, 1073, 2024.
    \bibitem{Keown} R.~Keown,  {Energy bands in diamond}, Physical Review 150.2, 568--573, 1966.%
    \bibitem{Kitaev} A.~Kitaev, {Anyons in an exactly solved model and beyond}, Annals of Physics Volume 321, 2--111, Issue 1, 2006.
     \bibitem{Lieb}E.~H.~Lieb, {Flux Phase of the Half-Filled Band}, Phys. Rev. Lett. 73, 2158--2161, 1994.
     \bibitem{Lawson} H.~B.~Lawson and  M.~L.~Michelsohn,  {Spin Geometry. Princeton University Press}, 1989.
     \bibitem{Kuchment2016} P.~Kuchment, {An overview of periodic elliptic operators}, {Bull. (New Ser.)} Amer. Math. Soc. Volume 53, Issue 3, 343--414, 2016.
    \bibitem{Keeffe} T.~O'Keeffe, {N-dimensional diamond, sodalite and rare sphere packings}, Actsa Crystalloger A 47:748--753, 1991.  
    \bibitem{Shinsei} S.~Ryu, {Three-dimensional topological phase on the diamond lattice}, Phys. Rev. B 79, 075124-1--9, 2009.%
    \bibitem{Sunada} T.~Sunada, {\it Topological Crystallography: With a View Towards Discrete Geometric Analysis}, Springer, 2012.
    \bibitem{Takahashi}R.~ Takahashi and S.~ Murakami, {Completely flat bands and fully localized states on surfaces of anisotropic diamond-lattice models}, Physical Review B—Condensed Matter and Materials Physics 88.23, 235303-1--10, 2013.%
    \bibitem{2}M.~Tsuchiizu, {Three-dimensional higher-spin Dirac and Weyl dispersions in the strongly isotropic $K_4$ crystal}, {Phys. Rev.B \bf{94}}, 195426-1--8, 2016.
    
   \end{thebibliography}
\end{document}